\documentclass[%
 reprint, 
superscriptaddress,
 amsmath,amssymb,
 aps, physrev,
]{revtex4-2}

\usepackage{epsfig,amssymb,amsmath,mathrsfs,amsthm,graphicx,float,color,hyperref,physics,enumitem}
\usepackage{qcircuit}
\usepackage[caption=false]{subfig}
\usepackage{algorithm}
\usepackage{algpseudocode}

\def\<{\langle}\def\>{\rangle}

\theoremstyle{plain}
\newtheorem{theorem}{Theorem}
\newtheorem{lemma}{Lemma}

\theoremstyle{definition}

\theoremstyle{remark}

\definecolor{gold}{rgb}{0.85, 0.65, 0.13}
\usepackage[dvipsnames]{xcolor}

\begin{document}

\title{Subsystem Quantum Error Correction for Noisy Quantum Metrology}

\author{Qiushi Liu}
 \email{qiushiliuph@gmail.com}
 \affiliation{Perimeter Institute for Theoretical Physics, Waterloo, Ontario N2L 2Y5, Canada}
\author{Sisi Zhou}
 \email{sisi.zhou26@gmail.com}
\affiliation{Perimeter Institute for Theoretical Physics, Waterloo, Ontario N2L 2Y5, Canada}
\affiliation{Department of Physics and Astronomy, Department of Applied Mathematics and Institute for Quantum Computing, University of Waterloo, Ontario N2L 3G1, Canada}

\begin{abstract}
   Quantum error correction has been successfully applied to enhance the precision of parameter estimation in the presence of noise. Nonetheless, existing methods require a number of noiseless, controllable ancillae and lack efficient encoding and decoding procedures. In this Letter, we demonstrate that subsystem error correction provides a new direction that can substantially simplify the metrological protocol. We derive general conditions under which subsystem stabilizer codes achieve the Heisenberg limit and show that, for broad classes of noise, this can be realized by syndrome-free protocols using at most a single ancilla qubit. Furthermore, we extend this framework to dynamical error correction and show that Floquet codes can protect time-dependent metrological signals in reaching the Heisenberg limit.

\end{abstract}

\maketitle

\emph{Introduction}---Quantum metrology \cite{Giovannetti2004Science,Giovannetti2006PRL,Giovannetti2011NPhoton,Degen2017RMP,Pezze2018RMP,Braun2018RMP} is the study of measuring physical parameters with high precision beyond classical limits. Given $N$ queries to a parameter-encoding quantum channel $\mathcal E_\theta$, the ultimate precision limit $\delta \hat{\theta}^2$ for estimating $\theta$ can scale as $1/N^2$, known as the Heisenberg limit (HL) \cite{Caves1981PRD,Holland1993PRL}, which surpasses the classical standard quantum limit (SQL) $1/N$ \cite{Caves1980RMP}. However, realistic physical systems are inevitably subject to noise, which typically degrades quantum coherence and restricts the achievable precision to SQL-like scaling \cite{Fujiwara2008JPA,Escher2011NP,Demkowicz-Dobrzanski2012NC,Demkowicz-Dobrzanski2014PRL}. This fundamental limitation has motivated the development of quantum error correction (QEC)-assisted metrology \cite{Arrad2014PRL,Kessler2014PRLQuantumErrorCorrection,Ozeri2013Heisenberg,Duer2014PRL,Demkowicz-Dobrzanski2017PRX,Zhou2018NC,Layden2018npjqi,Layden2019PRL,Chen2024PRL,Mann2025PRXQ,Antu2025Quantum}, which aims to suppress noise while preserving the coherent accumulation of the signal.

In QEC, logical information is redundantly encoded into a larger Hilbert space of physical qubits, enabling error detection and recovery procedures that protect the logical state against noise. The Knill–Laflamme conditions \cite{Knill1997PRA} provide necessary and sufficient criteria for a set of errors to be correctable by a given code. However, QEC-assisted quantum metrology demands additional care, since the signal itself must be meanwhile reserved by the error correction procedure. It has been shown that the HL can be restored using QEC when the ``Hamiltonian‑not‑in‑Kraus‑span'' (HNKS) condition \cite{Zhou2021PRXQ} is satisfied, which ensures the existence of a code supporting a nontrivial logical evolution. Despite its great success, existing theory of QEC for metrology typically assumes noiseless, controllable ancillae for constructing the code when the HNKS is satisfied. Moreover, detecting and correcting errors requires reliable syndrome measurements, which generally involve complex mid-circuit feedback and additional accurately prepared ancillae for readout. Therefore, experimental demonstration of QEC-assisted metrology that attains the HL remains a significant challenge \cite{Unden2016PRL,Wang2022NC}.

Here, we propose a strategy for metrology based on subsystem QEC \cite{Kribs2005PRL,Kribs2006Operator,Poulin2005PRL,Nielsen2007PRA} that substantially relaxes the requirements of conventional approaches. By encoding information in a logical subsystem while allowing noise to be absorbed into gauge degrees of freedom \cite{Poulin2005PRL}, subsystem codes enable simpler recovery procedures with reduced ancilla overhead, making experimental implementation more feasible. Working within the stabilizer formalism \cite{Gottesman1997Stabilizer}, we identify general conditions under which subsystem stabilizer codes can achieve the HL. We then derive sufficient conditions for achieving the HL with syndrome-free codes, using either no ancilla qubits or a single ancilla qubit. Extending beyond static codes, we further demonstrate that dynamical Floquet codes \cite{Hastings2021Quantum} can encode a time‑varying Hamiltonian and reach the HL. Our results establish subsystem QEC as a practical and resource-efficient paradigm for noisy quantum metrology.

\emph{Setting}---We consider the estimation of an unknown parameter $\theta$ using $N$ queries to a parameterized quantum channel $\mathcal E_\theta(\cdot)=\sum_{i=1}^r K_i^\theta (\cdot) K_i^{\theta\dagger}$. In a prototypical setting, such as noisy phase estimation, the Kraus operators can be factorized as 
\begin{equation}\label{eq:Kraus operators}
    K_i^\theta = \sqrt{p_i}E_ie^{-\mathrm i\theta H},
\end{equation} 
where $H$ is the signal Hamiltonian and $\{E_i\}$ is a set of noise operators occurring with probability distribution $\{p_i\}$. We assume that $\{E_i\}$ are multi-qubit Pauli operators to facilitate the use of the stabilizer formalism. We focus on local estimation around a reference value $\theta_0$, which can be set to $\theta_0=0$ without loss of generality by an appropriate unitary control.

Given a single-parameter quantum state $\rho_\theta$, the mean squared error (MSE) $\delta \hat{\theta}^2:=\mathbb E\left[(\hat{\theta}-\theta)^2\right]$ for any unbiased estimator $\hat{\theta}$ is fundamentally constrained by the quantum Cramér-Rao bound \cite{Helstrom1976Quantum,Holevo2011Probabilistic,Braunstein1994PRL} $\delta \hat{\theta}^2 \ge 1/[\nu F^Q(\rho_\theta)]$, where $F^Q(\rho_\theta)$ is the quantum Fisher information (QFI) and $\nu$ is the number of experimental repetitions. As the bound is saturable in the limit of $\nu \to \infty$, the QFI serves as the central figure of merit for single-parameter quantum metrology, quantifying the sensitivity of $\rho_\theta$ to small changes in $\theta$.

In quantum channel estimation, the goal is thus to design a protocol that yields an output state $\rho_\theta$ with high QFI. If $\mathcal E_\theta$ is noiseless, the optimal scaling of QFI is the HL $F^Q(\rho_\theta)\propto N^2$ by, e.g., repeatedly applying $\mathcal E_\theta$ $N$ times to an equal superposition of eigenstates of $H$ with minimal and maximal eigenvalues \cite{Giovannetti2006PRL}. 
However, the performance of such strategies is generally not robust in the presence of noise. HL can be restored using QEC provided that the HNKS condition \cite{Zhou2021PRXQ} is satisfied: $H\notin \mathsf{Span}\{K_i^{\theta \dagger}K_j^\theta, \forall i,j\}$, which guarantees the existence of a QEC code under which the noisy channel $\mathcal E_\theta$ combined with encoding and recovery procedures simulates a nontrivial logical unitary evolution in the code space.

\emph{Subsystem stabilizer codes}---The stabilizer formalism provides a compact framework for describing QEC codes \cite{Gottesman1997Stabilizer}. The central tool is the Pauli group $\mathcal P_n$ on $n$ qubits, which consists of all $n$-fold tensor products of single-qubit Pauli operators with phases $\{\pm 1,\pm \mathrm i\}$. An $[n,k]$ stabilizer code is given by the $+1$ eigenspace $\mathcal H_C$ of an abelian subgroup $\mathcal S=\langle S_1,\dots,S_s\rangle\subseteq \mathcal P_n$ with $s=n-k$ independent generators. Up to a Clifford transformation, the set of generators is isomorphic to single-qubit Pauli operators $\{Z_1',\dots,Z_s'\}$, and the code space is described by $k$ virtual qubits with logical Pauli operators $\{Z_i',X_i'\}_{i=s+1}^{n}$, which satisfy the canonical commutation relations: each pair $(Z_i',X_i')$ anticommutes, while operators belonging to different pairs commute. In this representation, $\mathcal P_n=\langle \mathrm i, Z_1',\dots,Z_n', X_1',\dots,X_n'\rangle$, where $\mathrm i$ denotes $\mathrm i I^{\otimes n}$ for simplicity.

By introducing the code space, the full Hilbert space is decomposed into the code space and its orthogonal complement, $\mathcal H=\mathcal H_C \oplus \mathcal H_C^{\perp}$. In subsystem QEC, the code space $\mathcal H_C$ is further factorized as $\mathcal H_C=\mathcal H_L \otimes \mathcal H_G$, where $\mathcal H_L$ encodes the logical information and $\mathcal H_G$ represents a gauge subsystem. In the stabilizer formalism for subsystem codes \cite{Poulin2005PRL}, this subsystem code structure is specified by a \emph{gauge group} $\mathcal G\subseteq \mathcal N(\mathcal S)$, the normalizer of $\mathcal S$, given by $\mathcal N(\mathcal S)=\langle \mathrm i, Z_1',\dots, Z_n', X_{s+1}', \dots, X_n' \rangle$. $\mathcal G$ can be generated by a set of \emph{check operators} that commute with stabilizers
\begin{equation}
    \mathcal G=\langle \mathrm i, Z_1', \dots, Z_s',Z_{s+1}',\dots,Z_{s+g}',X_{s+1}', \dots, X_{s+g}' \rangle
\end{equation}
for $g \le n-s$. The stabilizer group is simply the center of the gauge group, and the stabilizer information can be retrieved by measuring check operators. The logical group $\mathcal L$ is isomorphic to the quotient group $\mathcal N(\mathcal S)/\mathcal G$: 
\begin{equation}
    \mathcal N(\mathcal S)/\mathcal G \cong \mathcal L=\langle Z_{s+g+1}', \dots, Z_n', X_{s+g+1}', \dots, X_n'\rangle.
\end{equation}
As $\mathcal G$ and $\mathcal L$ commute, they induce a tensor product structure \cite{Zanardi2004PRL} on $\mathcal H_C=\mathcal H_L \otimes \mathcal H_G$ such that $\mathcal G$ acts nontrivially only on $\mathcal H_G$ and $\mathcal L$ acts nontrivially only on $\mathcal H_L$. This defines a subsystem stabilizer code $[n,k,g]$ with $k=n-s-g$ logical qubits. For convenience we also use the tuple $(\mathcal S,\mathcal G)$ to specify a subsystem stabilizer code.

For Pauli noise $\{{E_a}\subseteq \mathcal P_n\}$, a necessary and sufficient condition for correctability via subsystem stabilizer codes \cite{Poulin2005PRL} is
\begin{equation}
E_a E_b \notin \mathcal N(\mathcal S)-\mathcal G,
\end{equation}
which is a direct consequence of the Knill-Laflamme condition for general subsystem codes \cite{Kribs2005PRL,Kribs2006Operator,Nielsen2007PRA}: 
$\Pi_C E_a^\dagger E_b \Pi_C =  I_L\otimes g^{ab}_G$,
where $\Pi_C$ is the projection on the code space and $g^{ab}_G$ is an arbitrary operator acting on the gauge subsystem.

It is noteworthy that, in subsystem codes, the simplification of the recovery procedure is twofold: partitioning the encoded degrees of freedom into logical and gauge subsystems reduces the number of required stabilizer measurements, and syndrome information can be extracted via check operator measurements that typically involve lower-weight operators  \cite{Bacon2006PRA,Bravyi2013QIC}.

\emph{Effective dynamics by QEC}---A central mechanism enabling the recovery of HL using QEC is the realization of an effective logical unitary evolution. In Ref.~\cite{SM_note}, we show that for correctable Pauli errors, QEC projects the noisy evolution onto an effective Hamiltonian acting within the code space.

\begin{lemma} \label{lem:effective dynamics informal} (Informal)
    For a probe state $\rho$ supported on the code space with projection $\Pi_C$ and recovery operation $\mathcal R$, the resulting evolution $\mathcal R\circ \mathcal E_\theta(\rho)$ is governed by an effective Hamiltonian $H_{\mathrm{eff}}$ in the vicinity of $\theta=0$.

    In particular, for subsystem stabilizer codes, the recovery operation generally consists of standard syndrome decoding together with resetting the gauge subsystem. The syndrome extraction can also be constructed by measuring check operators.
\end{lemma}

Formal statements and proofs are provided in Lemmas~\ref{lem:effective dynamics stabilizer code}, \ref{lem:subsystem code dynamics} and \ref{lem:subsystem code dynamics measuring gauge} in Ref.~\cite{SM_note}. We highlight that subsystem codes for metrology combine standard decoding and gauge reset. While the latter is unnecessary for static quantum memory, it is usually essential here to ensure coherent accumulation under a nontrivial effective logical Hamiltonian $H_{\mathrm{eff}}$, unless $H_{\mathrm{eff}}$ acts trivially on the gauge subsystem. Moreover, we show that measuring check operators that are potentially easier does not degrade the metrological performance and can therefore be employed whenever experimentally convenient. Note that, similar to previous QEC sensing protocols, we also assume noiseless quantum operations, and our approach does not guarantee full fault tolerance.

\emph{Achieving the HL via subsystem codes}---The key to achieving the HL with $\rho_\theta=(\mathcal R\circ \mathcal E_\theta)^N(\rho)$ is using a code that corrects all the errors and preserves the signal through a nontrivial effective logical unitary. Formally we have (see Ref.~\cite{SM_note} for the proof):
\begin{theorem} \label{thm:subsystem code HL condition}
    Consider $N$ queries to $\mathcal E_\theta$ defined by Eq.~(\ref{eq:Kraus operators}). Let the Hamiltonian $H=\sum_k \lambda_k Q_k$ be decomposed in the Pauli basis $\{Q_k\}$. A subsystem stabilizer code $(\mathcal S,\mathcal G)$ achieves the HL if: 
    \begin{enumerate}[label=(\roman*)]
    \item $E_i E_j \notin \mathcal N(\mathcal S)-\mathcal G,\ \forall i,j \in \{1,\dots,r\}$;
    \item there exists at least one Pauli component $Q_{k_0} \in \mathcal N(\mathcal S) -\mathcal G$ such that 
    \begin{equation}
        \sum_{j:Q_j Q_{k_0} \in\mathcal S} \lambda_j \neq 0,
    \end{equation}
    where the sum runs over all Pauli components equivalent to $Q_k$ modulo stabilizers.
\end{enumerate}
\end{theorem}

Theorem~\ref{thm:subsystem code HL condition} provides a criterion for achieving the HL with a given subsystem code. Condition (i) is simply the Knill-Laflamme condition, and Condition (ii) ensures that the net contribution of all equivalent logical operators $Q_j$ differing by stabilizers is not vanishing. 

Furthermore, the following theorem gives a sufficient condition for the existence of such HL-achieving subsystem codes:

\begin{theorem}\label{thm:HL syndrome free}
    When $H\notin\mathsf{Span}\{\langle E_1,\dots,E_r \rangle\}$, the HL can be achieved by a syndrome-free subsystem stabilizer code $(\{I\},\mathcal G)$ with at most one noiseless ancillary qubit. The ancilla can be omitted, if, moreover, 

    \begin{enumerate}[label=(\roman*)]
        \item the decomposition of $H=\sum_k\lambda_kQ_k$ contains a Pauli term $Q_{k_0} \notin \mathsf{Span}\{\langle E_1,\dots,E_r \rangle\}$ such that $[Q_{k_0},E_i]=0,\ \forall i=1,\dots,r$; or
        \item the number of independent generators for a maximal abelian subgroup (modulo phase) of $\langle E_1,\dots,E_r \rangle$ is at most $n-1$.
    \end{enumerate} 
\end{theorem}

The proof is provided in Ref.~\cite{SM_note}. As a sketch, since there exists a Pauli component $Q_{k_0}$ of $H$ outside the noise algebra, by a Clifford transformation the noise generators can be mapped to a canonical form acting only on a set of gauge qubits, while $Q_{k_0}$ is mapped to an operator acting nontrivially on an additional logical qubit. If $Q_{k_0}$ commutes with the noise algebra, the logical qubit is directly decoupled from the gauge subsystem. Besides, the ancilla is only required when the noise algebra already occupies all $n$ qubits. Finally, we remark that, when Condition~(i) of Theorem~\ref{thm:HL syndrome free} is satisfied, our protocol avoids gauge reset and coincides with a previous result \cite[Theorem~3]{Liu2025PRL} obtained via a different proof technique \cite{previous_comparison_note}.  

The gauge reset must be performed in the tensor product basis determined by the gauge group, which can be implemented by applying Clifford circuits that transform between the computational basis and the tensor product basis. 
Alternatively, it can also be achieved by measuring corresponding check operators and applying gauge bit flips conditioned on check outcomes. Here for syndrome-free code, we can take the former approach to circumvent explicit measurements for gauge reset, which may be advantageous when readout noise is a dominant source of errors in quantum devices \cite{Arute2019Nature,Bengtsson2024PRL}. 

\emph{Phase estimation under bit flip noise}---
We begin with a simple illustrative example that provides insight into our approach. Consider single-qubit phase estimation in the presence of bit flip noise, where $U_\theta=e^{-\mathrm i\theta H}$ for $H=Z$ and $E_1=I,E_2=X$. By Theorem~\ref{thm:HL syndrome free}, there exists a syndrome-free subsystem code that achieves the HL, using one noiseless ancilla qubit. 

A single step of our syndrome-free QEC protocol is illustrated in Fig.\ref{fig:QEC bit flip noise}\subref{subfig:QEC single-qubit bit flip noise}. We denote by $U_\mathrm{CNOT}^{2\to 1}$ the CNOT gate taking the second qubit as the control and the first qubit as the target. Under the Clifford transformation $U_\mathrm{CNOT}^{2\to 1}(\cdot)(U_\mathrm{CNOT}^{2\to 1})^\dagger$, $Z_1,X_1$ are mapped to $Z_1Z_2,X_1$, respectively. The noise therefore acts only on the first qubit, which we identify as the gauge qubit, while the second qubit serves as the logical qubit. Choosing the gauge and logical states as $\ket{\psi_G}=\ket{0}$ and $\ket{\psi_L}=\ket{+}$, the probe state is $\ket{\psi}=(U_\mathrm{CNOT}^{2\to 1})^\dagger\ket{0}\ket{+}=\frac{1}{\sqrt{2}}(\ket{00}+\ket{11})$ in the computational basis. The recovery map is simply a gauge reset operation in the transformed basis. In contrast, the standard QEC protocol for metrology \cite{Kessler2014PRLQuantumErrorCorrection} requires mid-circuit feedback measurements and recovery (see Fig.~\ref{fig:QEC bit flip noise}\subref{subfig:conventional QEC single-qubit bit flip noise}).

\begin{figure} [!htbp]
    \centering
    \subfloat[Single-qubit syndrome-free scenario.]{\mbox{\Qcircuit @C=1em @R=1.2em {
     &\gate{\mathcal E_\theta}& \targ & \gate{\text{Reset}} & \targ & \qw \\
     &\qw & \ctrl{-1} & \qw & \ctrl{-1} & \qw 
    }
    }
    \label{subfig:QEC single-qubit bit flip noise}}\quad
    \subfloat[Single-qubit standard QEC scenario.]{\mbox{\Qcircuit @C=1em @R=1.2em {
     &\gate{\mathcal E_\theta}& \ctrl{2} & \qw & \gate{X} & \qw \\
     &\qw & \qw & \ctrl{1} & \qw & \qw \\
     &\qw & \targ & \targ & \meter \cwx[-2] &
    }
    }
    \label{subfig:conventional QEC single-qubit bit flip noise}}
    \\
    \subfloat[Multi-qubit syndrome-free scenario.]{\mbox{
    \Qcircuit @C=1em @R=1.5em {
     &\gate{\mathcal E_\theta}& \targ & \qw & \cdots & & \qw & \gate{\text{Reset}} & \qw & \cdots & & \targ & \qw \\
     & \vdots &  & & \cdots & &      & \vdots &      & \cdots &  &  & \\
     &\gate{\mathcal E_\theta}& \qw & \qw & \cdots & & \targ & \gate{\text{Reset}} & \targ & \cdots & & \qw & \qw \\
     &\qw & \ctrl{-3} & \qw & \cdots & &\ctrl{-1} & \qw & \ctrl{-1} & \cdots & & \ctrl{-3} & \qw 
    }
    }
    \label{subfig:QEC multi-qubit bit flip noise}} 
    \caption{\label{fig:QEC bit flip noise} Circuit diagram for one step of QEC in phase estimation under bit flip noise. In (a) and (c), the reset operation converts the state back to $\ket{0}$. The qubit on the last line represents the noiseless ancillary qubit, serving as the logical qubit. The qubits above it are gauge qubits. In (b), the state of the first two qubits is $(\ket{00}+\ket{11})/\sqrt{2}$. The qubit on the last line records the $Z_1Z_2$ syndrome that determines whether or not an $X_1$ correction is applied.}
\end{figure}
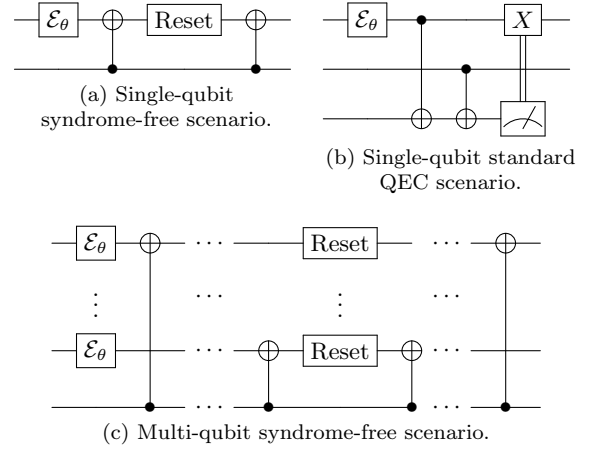

The QEC procedure introduced above naturally generalizes to the scenario of multiple qubits. Consider a system of $n_q$ qubits, each subject to independent single-qubit bit flip noise, and the task is to estimate $\theta$ using $N$ queries to $\mathcal E_\theta^{\otimes n_q}$. We summarize the procedure as follows, with a single step of the multi-qubit protocol illustrated in Fig.~\ref{fig:QEC bit flip noise}\subref{subfig:QEC multi-qubit bit flip noise}:
\begin{enumerate}[label=(\arabic*)]
    \item Prepare $\ket{\psi}=\frac{1}{\sqrt{2}}\left(\ket{0}^{\otimes (n_q+1)}+\ket{1}^{\otimes (n_q+1)}\right)$.
    \item Apply $\mathcal E_\theta^{\otimes n_q}$ to the first $n_q$ qubits.
    \item For each qubit $i=1,\dots,n_q$, perform the following recovery sequence: apply $U_\mathrm{CNOT}^{\,n_q+1 \to i}$, reset the $i$-th qubit to $\ket{0}$, and then apply $U_\mathrm{CNOT}^{\,n_q+1 \to i}$ again.
    \item Repeat steps (2)--(3) a total of $N$ times.
    \item Measure all the qubits in the $X$ basis.
\end{enumerate}

This example also highlights a distinction between our approach and decoherence-free subspace \cite{Lidar1998PRL} methods for metrology \cite{Layden2018npjqi,Liu2025PRL,Pereira2023SR}. The latter can be viewed as a special case where the signal is protected passively, while active gauge reset is required here to handle more general classes of noise as illustrated by this example.

\emph{Dynamical QEC}---We now extend the discussion to a dynamical QEC setting for the purpose of sensing periodic quantum signals. Floquet codes \cite{Hastings2021Quantum,Davydova2023PRXQuantum,Davydova2024Quantum,TownsendTeague2023FloquetifyingColourCode,Alam2025PRA,Eickbusch2025NP}, as a variant of subsystem codes, feature dynamically generated logical qubits throughout a measurement cycle of noncommuting check operators, which combine across the cycle to give deterministic syndrome outcomes. Here we focus on a specific type of Floquet codes, named the ladder code \cite{Hastings2021Quantum}, to illustrate the principle of Floquet code metrology. We anticipate that the generalization to other classes of Floquet codes is straightforward, whenever the signal can be encoded in a periodic logical Hamiltonian.
The ladder code architecture consists of two horizontal legs connected by vertical rungs, as shown in Fig.~\ref{fig:ladder}. The left and right ends of the ladder are connected to form a torus. Denoting the number of plaquettes by $n_p$ (assumed to be an even number), the total number of physical qubits is $2n_p$. One can perform alternating $XX$ and $YY$ checks on top and bottom horizontal legs and $ZZ$ checks on vertical rungs. It can be verified that the gauge group generated by all check operators is full-rank, meaning no logical information is encoded statically.  

\begin{figure} [!htbp]
    \centering
    \includegraphics[width=0.45\textwidth]{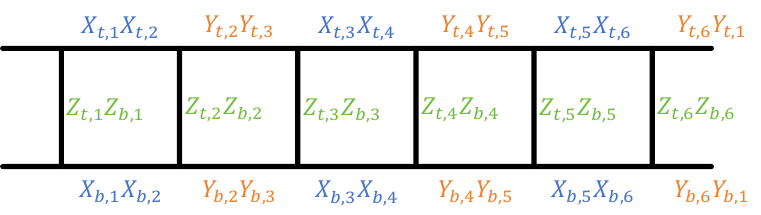}
    \caption{\label{fig:ladder}Ladder code architecture (for $n_p=6$). Qubits lie on the vertices and check operators on two neighboring qubits are measured.}
\end{figure}

A cycle of syndrome measurements consists of four rounds of checks: 
\begin{equation}
    \begin{cases}
        \text{round }0 \pmod 4: ZZ\text{ checks},\\
        \text{round }1 \pmod 4: XX\text{ checks},\\
        \text{round }2 \pmod 4: ZZ\text{ checks},\\
        \text{round }3 \pmod 4: YY\text{ checks}.\\
    \end{cases}
\end{equation}
 
After each round of checks, the measured check operators (up to signs) together with all four-qubit plaquette operators (i.e., the products of checks around each square) generate an instantaneous stabilizer group consisting of $2n_p-1$ independent stabilizer elements, thus inducing one logical qubit. The syndrome bits are eigenvalues of plaquette stabilizers, inferred from pairs of two recent consecutive measurement rounds. These syndrome bits always take the value $+1$ in the absence of noise, while any flips indicate faults and can be used to infer the most likely error configuration.

We can choose an inner logical operator
\begin{equation}
    Z_L=Z_{b,1}Z_{b,2}\cdots Z_{b,n_p}
\end{equation}
which commutes with all check operators throughout the measurement cycle. The dynamical nature of the code is reflected in the outer logical operator on any vertical rung
\begin{equation}
    X_L=\begin{cases}
        \pm X_{t,i}X_{b,i}, \text{ after round 1},\\
        \pm Y_{t,i}Y_{b,i}, \text{ after round 3}, \\
        \pm X_{t,i}X_{b,i}\text{ or }\pm Y_{t,i}Y_{b,i}, \text{ after round 0, 2},
\end{cases}
\end{equation} 
for any $i$. The sign has to be determined by the history of check measurement outcomes. Apparently, $X_L$ anticommutes with $Z_L$ and commutes with all the instantaneous stabilizers. Therefore, $X_L$ and $Z_L$ define a dynamical logical qubit, with the logical subspace tracked through the recorded check measurement outcomes. As noted in Ref.~\cite{Hastings2021Quantum}, the ladder code can correct a sequence of stochastic Pauli errors characterized by Kraus operators proportional to check operators (which we call check operator noise), in a low error rate regime. We defer the detailed discussions of its error correctability to End Matter.

\emph{Floquet codes for time-dependent metrology}---In a time-dependent metrological setting, it is natural to encode the signal of interest in the dynamically evolving logical information. Consider a time-dependent Hamiltonian $H(t)=H^{(k)}$ for each time interval $t\in[k\Delta t,(k+1)\Delta t),\ k=0,\dots,N-1$, which changes alternatively between sums of vertical rung operators multiplied by the parameter of interest $\theta$:
\begin{equation}
    H^{(k)}=\begin{cases}
        \theta\sum_{i=1}^{n_p} X_{t,i}X_{b,i},\text{ for odd $k$},\\
        \theta \sum_{i=1}^{n_p} Y_{t,i}Y_{b,i},\text{ for even $k$}.
    \end{cases}
\end{equation}
We can synchronize the measurement cycle of the ladder code to protect the logical evolution of the state. Specifically, we start with $ZZ$ checks at $t=\Delta t/2$, and perform $ZZ\to XX\to ZZ\to YY$ checks periodically with a time interval $\Delta t/2$ between neighboring rounds.
The errors are assumed to be products of check operators that occur right before each measurement (unlike conventional single-qubit error models) and are sparse enough such that logical errors are highly unlikely. All operations (measurements, state preparation and controls) are assumed to be instantaneous.

We formulate the metrological protocol in Algorithm~\ref{alg:floquet metrology}, and provide more explanations in End Matter. 

\begin{algorithm}[H]
\caption{Floquet code metrology protocol.}
\label{alg:floquet metrology}
\begin{algorithmic}[1]
\State Prepare an initial logical state stabilized by all plaquette stabilizers and by the $+1$ $YY$ checks:
\begin{equation}\label{eq:Floquet code probe}
\ket{\psi}=\frac{1}{\sqrt{2}}\left(\ket{\lambda_{\mathrm{min}}}+\ket{\lambda_{\mathrm{max}}}\right),
\end{equation}
where $\ket{\lambda_{\mathrm{min}}}$ and $\ket{\lambda_{\mathrm{max}}}$ are eigenstates of $H^{(0)}$ with minimal and maximal eigenvalues.
\For{$r=0,\ldots,2N-1$}
    \State At time $t=(r+1)\Delta t/2$, perform round $r$ checks.
    \State Infer the error $E$ from plaquette syndromes obtained by checks of rounds $r$ and $r-1$ (assuming $+1$ $YY$ checks at round $-1$), and apply the correction $E^\dagger$. Decode the error-free check outcomes. 
    \If{$r \equiv 0\pmod 4$}
    \If{a pair of neighboring error-free $Z_{t,i}Z_{b,i}$ and $Z_{t,i+1}Z_{b,i+1}$ checks are both $+1$ for any even $i$}
        \State Apply a $Y_{b,i}Y_{b,i+1}$ unitary control.
    \EndIf
    \ElsIf{$r \equiv 2\pmod 4$}
    \If{a pair of neighboring error-free $Z_{t,j}Z_{b,j}$ and $Z_{t,j+1}Z_{b,j+1}$ checks are both $+1$ for any odd $j$}
        \State Apply an $X_{b,j}X_{b,j+1}$ unitary control.
    \EndIf
    \EndIf
\EndFor
\State Measure in the basis of the logical operator $Z_L$.
\end{algorithmic}
\end{algorithm}

For a low error rate the dynamics can be approximated as an effectively unitary evolution in a periodically changing reference basis, allowing the QFI to achieve the HL $F^Q(\rho_\theta)\propto N^2$.

\emph{Discussion}---We have introduced a framework for quantum metrology based on subsystem stabilizer codes, demonstrating that the HL can be achieved under resource requirements far less stringent than those required by standard QEC-assisted protocols. The core principle is that protecting the full quantum state is unnecessary for metrological purposes, while it suffices to preserve coherent evolution within a logical subsystem. This insight facilitates recovery protocols that avoid explicit syndrome extraction and reduce the ancilla overhead. We have also shown the advantage of  dynamical QEC for time-dependent metrology by performing check measurements that enter the instantaneous stabilizer group in a cycle.

We remark that, in terms of syndrome-free subsystem codes, our approach aligns with the operational philosophy of autonomous error correction for metrology \cite{Kwon2025Restoring,Ni2025Autonomous}, in which engineered dissipation removes the error entropy to a heat bath without explicit feedback measurements. 
However, a crucial distinction is that we do not assume access to reservoir engineering that implicitly performs syndrome monitoring; instead, the gauge subsystem naturally absorbs the error entropy, leaving the logical subsystem coherent and preserving the metrological signal. We anticipate that this architecture offers a promising pathway toward resource-efficient noisy quantum sensors.

\emph{Acknowledgments}---Q.L. and S.Z. acknowledge support from National Research Council of Canada (Grant No.~AQC-217-1) and Perimeter Institute for Theoretical Physics, a research institute supported in part by the Government of Canada through the Department of Innovation, Science and Economic Development Canada and by the Province of Ontario through the Ministry of Colleges and Universities. 

\section*{End Matter}
\emph{Mechanism of the Floquet code metrology protocol}---The principle here is to ensure $H(t)$ coherently accumulates the signal. To this end, by noting that $H(0)=n_p \theta X_{L,-1}$, we prepare a probe state Eq.~(\ref{eq:Floquet code probe}) in the code space to maximize $\mel{\psi}{X_{L,-1}^2}{\psi}-(\mel{\psi}{X_{L,-1}}{\psi})^2$ so that we can obtain the highest QFI \cite{Giovannetti2006PRL}. In all the rounds of check measurements, the transition between instantaneous stabilizer groups should guarantee that $H(t)=n_p \theta X_{L}$ all the time. This is achieved by the following procedure: in each round, we decode the error $E$ by minimum-weight perfect matching \cite{Edmonds1965CJM,Edmonds1965Maximum} from the check measurements and obtain the error-free outcomes. Unlike using the code as a static memory where we do not need real-time correction, here we have to apply explicit control to ensure that $H(t)=n_p \theta X_{L}$. By applying $E^\dagger$, the state is brought back to the error-free code space. By construction, a pair of error-free $ZZ$ checks that form a $+1$ plaquette stabilizer must have the same parity. Then, if all error-free $Z_{t,k}Z_{b,k}$ check outcomes are $-1$, we can freely choose $X_L=X_{t,k}X_{b,k}$ or $X_L=Y_{t,k}Y_{b,k}$ after the check measurements, by noting that $X_{t,k}X_{b,k}Y_{t,k}Y_{b,k}=-Z_{t,k}Z_{b,k}$. In this case, no additional control is required and $H(t)=n_p \theta X_{L}$ is naturally satisfied. However, if a pair of error-free $Z_{t,k}Z_{b,k}$ check outcomes are both $+1$ for any $k$, we need an additional control (as in Algorithm~\ref{alg:floquet metrology}) to rotate the code space to the one associated with $-1$ $ZZ$ checks.

\emph{Error correction of ladder codes}---Here we provide a detailed analysis of the logical state evolution in the ladder code and the effect of errors, and clarify the distinction between QEC for quantum memory and that for quantum metrology.

Let us consider an initial logical state before round $0$ (which is in the instantaneous stabilizer group after round $-1$ when we have all $YY$ checks of value $+1$):
\begin{equation}
    \ket{\psi_L}_{-1}=\alpha\ket{0_L}_{-1}+\beta\ket{1_L}_{-1},
\end{equation}
where $Z_L\ket{0_L}_{-1}=\ket{0_L}_{-1}$, $Z_L\ket{1_L}_{-1}=-\ket{1_L}_{-1}$, and $X_{L,-1}\ket{0_L}_{-1}=\ket{1_L}_{-1}$, having denoted $X_L$ after round $r$ by $X_{L,r}$. At round $0$, we measure $ZZ$ check operators and the postselected state is 
\begin{equation}
    \begin{aligned}
        \ket{\psi_L}_{0}&=\frac{1}{\sqrt{\Tr(\Pi_{\mathbf z}^{ZZ}\ket{\psi_L}_{-1}\bra{\psi_L}_{-1})}}\Pi^{ZZ}_{\mathbf z}\ket{\psi_L}_{-1}\\
        &=2^{\frac{n_p}{4}}\Pi^{ZZ}_{\mathbf z}\ket{\psi_L}_{-1},
    \end{aligned}
\end{equation}
given a string of check measurement outcomes $\mathbf z=(z_1,\dots,z_{n_p})$ obtained from the projection $\Pi^{ZZ}_\mathbf z:=\prod_i\frac{1}{2}(I+(-1)^{z_i}Z_{t,i}Z_{b,i})$. The prefactor stems from the fact that $2^{n_p/2}$ $ZZ$ check outcomes are equally probable, constrained by $+1$ plaquette stabilizers. We can regard this state evolution as a change of frame for the logical subspace, by defining
\begin{equation}
    \ket{0_L}_0=2^{\frac{n_p}{4}}\Pi^{ZZ}_{\mathbf z}\ket{0_L}_{-1}
\end{equation}
and 
\begin{equation}
    \ket{1_L}_0=2^{\frac{n_p}{4}}\Pi^{ZZ}_{\mathbf z}\ket{1_L}_{-1}.
\end{equation}
Similar procedures can be performed in subsequent rounds to obtain $\ket{0_L}_r$ and $\ket{1_L}_r$. Fixing $r$, one can always require the logical operator representatives of $X_L$ after round $r-1$ checks and round $r$ checks to be the same, by adjusting one of them with a $ZZ$ check operator if distinct. This is due to the following observation: $ZZ$ checks are performed at either round $r-1$ or round $r$. If $ZZ$ checks are measured at round $r-1$, after the measurements $X_{L,r-1}$ can be chosen as either $\pm XX$ or $\pm YY$, since they differ by a $ZZ$ check operator up to a sign. By a proper choice of $X_{L,r-1}$, after round $r$ of $XX$ ($YY$) checks $X_{L,r}$ is the same as $X_{L,r-1}$. Similar arguments hold if $ZZ$ checks are performed at round $r$. This ensures $X_{L,r}\ket{0_L}_r=\ket{1_L}_r$. Therefore, the code state after round $r$ checks is preserved if we change the frame consistently:
\begin{equation}
    \ket{\psi_L}_{r}=\alpha\ket{0_L}_r+\beta\ket{1_L}_r.
\end{equation}

When a Pauli error occurs, it flips the outcomes of check operators that anticommute with it. Without loss of generality, let us assume that an $XX$ check operator error $E$ happens before round $0$ measurements. The noisy post-measurement state is
\begin{equation}
    \ket{\psi_L}_0^E=\Pi^{ZZ}_{\mathbf z_E}E\ket{\psi_L}_{-1} = E\Pi^{ZZ}_\mathbf z \ket{\psi_L}_{-1}=E\ket{\psi_L}_0,
\end{equation}
where $\Pi^{ZZ}_{\mathbf z_E}$ denotes the projection on noisy check outcomes $\mathbf z_E$ and $\Pi^{ZZ}_{\mathbf z}$ is the error-free check outcome projection. In traditional QEC for quantum memory, this check operator noise can thus be absorbed into $XX$ checks at later rounds without active decoding \cite{Hastings2021Quantum}, since after $XX$ check measurements associated with projection $\Pi^{XX}_\mathbf x$ at round $1$ the state becomes 
\begin{equation}    
    \begin{aligned}
        \Pi^{XX}_\mathbf x \ket{\psi_L}_0^E&=E\Pi^{XX}_\mathbf x\Pi^{ZZ}_\mathbf z \ket{\psi_L}_{-1}\\
        &\propto\Pi^{XX}_\mathbf x\Pi^{ZZ}_\mathbf z \ket{\psi_L}_{-1}\\
        &=\alpha\ket{0_L}_1+\beta\ket{1_L}_1.
    \end{aligned}
\end{equation} 
Similar arguments hold for $XX$ or $YY$ check operator errors occurring before any check measurements, while $ZZ$ check operator errors are trivial. Moreover, it is unnecessary to distinguish a check operator error on the top leg or the bottom leg, as they differ by a stabilizer during $ZZ$ checks. Therefore, a single check operator error has no long-term effect on the logical state and is automatically corrected for idle quantum memory. However, in metrology the errors may not commute with the sensing Hamiltonian and should be decoded and corrected in time after each round of checks, as in Algorithm~\ref{alg:floquet metrology}.

A product of multiple check operator errors can be harmful if it cannot be decoded correctly. The reason is that multiple $XX$ and $YY$ errors can combine to form a logical error, which can destroy the logical information. Fortunately, under a low error rate, the Floquet code can decode correctly by minimum-weight perfect matching. Assume that the single check operator error rate is $p$, and a logical error can only occur if at least $n_p/4$ check operator errors occur, otherwise we are able to decode perfectly. Therefore, if we perform active QEC reach round, the logical error rate is suppressed to $O(p^{n_p/4})$, decreasing exponentially with an increasing lattice size.

The ladder code can also tolerate low-rate classical measurement errors, as each measurement error flips a pair of syndromes in two consecutive rounds, and an odd number of flipped syndrome bits detect the measurement errors. A consecutive sequence of check operator errors and measurement errors can be decoded in the $(1+1)$-dimensional spacetime as an error chain with two nontrivial syndrome bits at the end. Therefore, a decoder for the two-dimensional toric code \cite{Kitaev2003AOP} can be used to decode the errors with high confidence. A subtle issue arises when decoding measurement errors that occur at the end of the error chain, as the location of the measurement errors becomes ambiguous. In quantum computing, this can be resolved fortunately by choosing a logical operator $X_L$ through majority decoding, as $X_L$ can be chosen at $n_p$ locations. However, the same approach becomes challenging in metrology, where the sensing Hamiltonian playing the role of the logical operator is already given in our setting.


\begin{thebibliography}{63}
\bibitem [{pre()}]{previous_comparison_note}%
  \BibitemOpen
  \href@noop {} {}\bibinfo {note} {While the gauge state is chosen to be a maximally mixed state in Ref.~\cite{Liu2025PRL}, the protocol here (using a pure probe state) is esstentially equivalent, as $H_{\mathrm{eff}}$ is independent of the gauge state when Condition (i) is satisfied.}\BibitemShut {Stop}%
\end{thebibliography}
\providecommand{\noopsort}[1]{}\providecommand{\singleletter}[1]{#1}%

\appendix
\onecolumngrid
\section{Formal statements and proofs of Lemma~\ref{lem:effective dynamics informal}} \label{app:lem:effective dynamics informal}
We begin with the effective dynamics using standard stabilizer codes for metrology. Below we show that, for correctable Pauli errors $\{E_i\}$, the state evolution via stabilizer codes is only governed by an effective Hamiltonian $H_{\mathrm{eff}}$ up to the first order of $\theta=\theta_0+d\theta=d\theta$, for $\theta_0=0$. 
\begin{lemma} (Effective dynamics via stabilizer codes.) \label{lem:effective dynamics stabilizer code}
    Consider a stabilizer code with code space projection $\Pi_C$ and a probe state supported on the code space, $\rho=\Pi_C \rho \Pi_C$. If $\{E_i\}_{i=1}^r$ are correctable Pauli errors satisfying $\Pi_CE_i^\dagger E_j \Pi_C \propto \Pi_C$, then there exists a recovery operation $\mathcal R$ such that for $\theta=d\theta$
    \begin{equation}
        \rho_{\theta}=\mathcal R\circ\mathcal E_\theta(\rho)=\rho-\mathrm i[H_{\mathrm{eff}}, \rho]d\theta+O(d\theta^2)
    \end{equation}
    for $H_{\mathrm{eff}}:=\Pi_CH\Pi_C$.
\end{lemma}

\begin{proof}
    For a stabilizer code, each error $E_i$ maps the code space to either identical or orthogonal subspaces. Denoting the syndrome of error $E_i$ by $\beta(i)$, we have
    \begin{equation}
        \Pi_{\beta(i)}\Pi_{\beta(j)}=\delta_{\beta(i)\beta(j)}\Pi_{\beta(i)}
    \end{equation}
    for projection $\Pi_{\beta(i)}= E_i\Pi_CE_i^\dagger$. The Kraus operators for the recovery operation $\mathcal R$ can be chosen as
    \begin{equation}
        R_m=U_m\Pi_m,\ m\in\{\beta(i)\mid i=1,\dots,r\}
    \end{equation}
    with unitary $U_m= E_j^\dagger$ for some $j$ such that $\beta(j)=m$. Such a choice ensures that $U_m \Pi_m E_i\Pi_C =c_{i}\delta_{m\beta(i)} \Pi_C$ for some scalar $c_i$. We denote the projection on the leakage error subspace by $\Pi_E:=I-\Pi_C$. After the recovery map the state becomes
    \begin{equation}
        \begin{aligned}
            \rho_\theta&=\rho-\mathrm i\sum_{im}R_mE_i[H,\rho]E_i^\dagger R_m^\dagger d\theta+O(d\theta^2)\\
            &=\rho-\mathrm i\sum_{im}R_mE_i[(\Pi_C+\Pi_E)H(\Pi_C+\Pi_E),\rho]E_i^\dagger R_m^\dagger d\theta+O(d\theta^2)\\
            &=\rho-\mathrm i[\Pi_CH\Pi_C,\rho]d\theta-\mathrm i \sum_{im}R_mE_i(\Pi_EH\rho-\rho H\Pi_E)E_i^\dagger R_m^\dagger d\theta+O(d\theta^2)\\
            &=\rho-\mathrm i[\Pi_CH\Pi_C,\rho]d\theta+O(d\theta^2),
        \end{aligned}
    \end{equation}
    where in the last equality we have used $\Pi_{\beta(i)}E_i\Pi_E=0$.
\end{proof}

For subsystem stabilizer codes, we similarly show that the effective evolution remains unitary within the code space $\mathcal H_C$ by resetting the gauge state. 

\begin{lemma} (Effective dynamics via subsystem stabilizer codes.)\label{lem:subsystem code dynamics}
    Consider a subsystem stabilizer code space $\mathcal H_C=\mathcal H_L \otimes \mathcal H_G$ associated with the projection $\Pi_C$. We take a probe state $\rho=\rho_L\otimes\rho_G$, where $\rho_L$ is supported on $\mathcal H_L$ and $\rho_G$ is an initial gauge state. If $\{E_i\}$ are correctable Pauli errors satisfying $\Pi_C E_i^\dagger E_j \Pi_C = I_L \otimes g^{ij}_G$, then there exists a recovery operation $\mathcal R$ such that for a small $\theta=d\theta$
    \begin{equation} \label{eq:subsystem code dynamics}
        \rho_{\theta}=\mathcal R\circ\mathcal E_\theta(\rho)=\left\{\rho_L-\mathrm i[H_{\mathrm{eff}}, \rho_L]d\theta+O(d\theta^2) \right\} \otimes \rho_G
    \end{equation}
    for $H_{\mathrm{eff}}:=\sum_k \langle H_{G,k} \rangle_{\rho_G} H_{L,k}$, where we have decomposed $\Pi_CH\Pi_C=\sum_kH_{L,k}\otimes H_{G,k}$ and used the notation $\langle H\rangle_{\rho}:=\Tr(H\rho)$. The recovery operation $\mathcal R$ includes resetting the gauge state to $\rho_G$.
\end{lemma}
\begin{proof}
    As in the proof of Lemma~\ref{lem:effective dynamics stabilizer code}, we can denote syndrome projections by $\Pi_{\beta(i)} = E_i\Pi_C E_i^\dagger$. The recovery operation $\mathcal R$ can be chosen as a composition of two processes $\mathcal R=\mathcal R^{(2)} \circ \mathcal R^{(1)}$. $\mathcal R^{(1)}$ plays the standard role of projection and correction and is represented by Kraus operators $R_m^{(1)}=U_m\Pi_m$ such that
    \begin{equation}
        R_m^{(1)}E_i\Pi_C=U_m\Pi_mE_i\Pi_C=\delta_{m\beta(i)}(I_L \otimes u_{G,i})\Pi_C
    \end{equation}
    for an arbitrary gauge operator $u_{G,i}$ satisfying $\sum_i u_{G,i}^\dagger u_{G,i}=I_G$. $U_m^\dagger= E_j$ for some $j$ such that $\beta(j)=\beta(i)$, which is equivalent to $E_i$ up to a gauge transform. Additionally, $\mathcal R^{(2)}(\cdot)=\Tr_G(\cdot)\otimes \rho_G$ is a reset operation that resets the state of the gauge subsystem. 
    
    We decompose the Hamiltonian projected on the code space as $\Pi_CH\Pi_C= \sum_k H_{L,k}\otimes H_{G,k}$. By denoting the Kraus operators of $\mathcal R$ by $\{R_m\}_m$, the recovered logical state $\rho^\theta_L=\Tr_G[\mathcal R\circ\mathcal E_\theta(\rho_L\otimes \rho_G)]$ is
    \begin{equation}
        \begin{aligned}
            \rho^\theta_L&=\rho_L-\mathrm i\sum_{im}\Tr_G\left(R_mE_i[H,\rho_L\otimes \rho_G]E_i^\dagger R_m^\dagger \right)d\theta+O(d\theta^2)\\
            &= \rho_L - \mathrm{i} \left[\sum_k \langle H_{G,k} \rangle_{\rho_G} H_{L,k}, \rho_L\right] d\theta +O(d\theta^2),\\
        \end{aligned}
    \end{equation}
    by noting that $\Pi_{\beta(i)}E_i\Pi_E=0$. Since the gauge subsystem is reset to the state $\rho_G$ by construction, we obtain the desired result. Note that the reset operation is typically necessary if we would like to apply the QEC procedure $N$ times to obtain a coherent logical signal accumulation $\rho_{L}^{\theta,(N)}:=\Tr_G\left[(\mathcal R\circ\mathcal E_\theta)^N(\rho)\right]=\rho_L - \mathrm{i} \left[NH_{\mathrm{eff}}, \rho_L\right] d\theta +O(d\theta^2)$.
\end{proof}

Note that in Lemma~\ref{lem:subsystem code dynamics} we have implicitly chosen a basis that admits the tensor product structure of logical and gauge subsystem. This basis is equivalent to the computational basis up to a Clifford unitary transformation that needs to be identified by the gauge group.

We furthermore remark that the results of Lemma~\ref{lem:effective dynamics stabilizer code} and \ref{lem:subsystem code dynamics} can be generalized to general QEC codes beyond the stabilizer formalism. This can be seen by choosing a proper Kraus representation for noise operators such that correctable errors produce orthogonal error subspaces \cite{Knill1997PRA}, and the same line of reasoning in the proofs above still holds.

Finally, as in subsystem codes for quantum computation, the syndrome extraction can also be obtained via check measurements for our metrological purpose.
\begin{lemma} (Validity of check measurements.) \label{lem:subsystem code dynamics measuring gauge}
    Under the assumptions of Lemma~\ref{lem:subsystem code dynamics}, the syndrome detection in recovery operation can be constructed by measuring check operators, while still resulting in the same unitary dynamics Eq.~\eqref{eq:subsystem code dynamics}.
\end{lemma}
\begin{proof}
    The syndrome detection and correction operation $\mathcal R^{(1)}$ in the proof of Lemma~\ref{lem:subsystem code dynamics} is replaced by $\mathcal R^{(1)\prime}$ represented by Kraus operators
    \begin{equation}
        R^{(1)\prime}_{m_1,\dots,m_{s+2g}}=U'_m\Pi'_{m_{s+2g}}\cdots \Pi'_{m_1}
    \end{equation}
    that sequentially projects onto the subspaces associated with check measurement outcomes and subsequently applies a recovery unitary $U'_m$ corresponding to the syndrome outcome. Specifically, $m=\mathsf s(m_1,\dots,m_{s+2g})$ is the syndrome determined by the sequence of check measurement outcomes $m_1,\dots,m_{s+2g}$. Since check measurements do not distort the logical subsystem, we have
    \begin{equation}
        U_{\mathsf s(m_1,\dots,m_{s+2g})}'\Pi'_{m_{s+2g}}\cdots \Pi'_{m_1}E_i\Pi_C=\delta_{\mathsf s(m_1,\dots,m_{s+2g})\beta(i)}\left[I_L\otimes (g'_G)^{m_1,\dots,m_{s+2g}}\right]\Pi_C
    \end{equation}
    for some check operator $(g'_G)^{m_1,\dots,m_{s+2g}}$, and $\beta(i)$ denotes the syndrome of error $E_i$. Note that
    \begin{equation}
        \Pi'_{m_{s+2g}}\cdots \Pi'_{m_1}E_i\Pi_E=0,\ \text{for }\mathsf s(m_1,\dots,m_{s+2g})=\beta(i).
    \end{equation}
    We still take $\mathcal R^{(2)}(\cdot)=\Tr_G(\cdot)\otimes \rho_G$ as the reset operation and $\mathcal R=\mathcal R^{(2)} \circ \mathcal R^{(1)}$. Therefore, representing $\mathcal R$ by Kraus operators $\{R_m\}_m$, we obtain the evolved logical state $\rho^\theta_L=\Tr_G[\mathcal R\circ\mathcal E_\theta(\rho_L\otimes \rho_G)]$:
    \begin{equation}
        \begin{aligned}
            \rho^\theta_L&=\rho_L-\mathrm i\sum_{im}\Tr_G\left(R_mE_i[H,\rho_L\otimes \rho_G]E_i^\dagger R_m^\dagger \right)d\theta+O(d\theta^2)\\
            &= \rho_L - \mathrm{i} \left[\sum_k \langle H_{G,k} \rangle_{\rho_G} H_{L,k}, \rho_L\right] d\theta +O(d\theta^2).\\
        \end{aligned}
    \end{equation}
\end{proof}

\section{Proof of Theorem~\ref{thm:subsystem code HL condition}}
\begin{proof}
    Condition~(i) is exactly the criterion for correctable errors \cite{Poulin2005PRL}. Lemma~\ref{lem:subsystem code dynamics} shows that we can choose a code state in the tensor product form $\rho=\rho_L\otimes \rho_G$ and a recovery operation $\mathcal R$ such that 
    \begin{equation}
        \dot \rho_\theta=\mathcal R\circ\dot{\mathcal E}_\theta(\rho)=-\mathrm i[H_{\mathrm{eff}}\otimes I_G,\rho_L\otimes\rho_G].
    \end{equation}
    Here $\dot{\ast}$ means the derivative of $\ast$ with respect to $\theta$. After $N$ repeated applications of the parameter encoding and recovery process, the derivative of the final output state $\rho_\theta^{(N)}=\left(\mathcal R\circ\mathcal E_\theta\right)^N(\rho)$ is
    \begin{equation}
        \dot \rho_\theta^{(N)}=\frac{d}{d\theta}\left(\mathcal R\circ\mathcal E_\theta\right)^N(\rho)=-\mathrm i[NH_{\mathrm{eff}}\otimes I_G,\rho_L\otimes\rho_G].
    \end{equation}
    
    Condition~(ii) guarantees that $H$ has a component playing the role of a nontrivial logical operator. In particular, for $Q_{k_0}\in\mathcal N(\mathcal S)-\mathcal G$, by a Clifford transformation and projection on the code space,
    \begin{equation}
        \Pi_C Q_{k_0} \Pi_C \cong   H'_{L,k_0} \otimes H'_{G,k_0},
    \end{equation}
    where the logical component $H'_{L,k_0} \not\propto I$. This implies
    \begin{equation}
        \sum_{j:Q_j Q_{k_0} \in\mathcal S} \lambda_j \Pi_C Q_j \Pi_C \cong \sum_{j:Q_j Q_{k_0} \in\mathcal S} \lambda _j  H'_{L,k_0} \otimes H'_{G,k_0},
    \end{equation}
    which is nontrivial if $\sum_{j:Q_j Q_{k_0} \in\mathcal S} \lambda _j \neq 0$. Taking into consideration the gauge freedom, all the terms associated with the same logical $H'_{L,k_0}$ amount to
    \begin{equation}
        \sum_{j:Q_j Q_{k_0} \in\mathcal G} \lambda_j \Pi_C Q_j \Pi_C \cong \sum_{j:Q_j Q_{k_0} \in\mathcal G} \lambda _j  H'_{L,k_0} \otimes H'_{G,j},
    \end{equation}
    To obtain a nontrivial effective Hamiltonian $H_{\mathrm{eff}}\not\propto I$, we would like to maximize the absolute value of $\sum_{j:Q_jQ_{k_0}\in\mathcal G}\lambda_j\langle H'_{G,j} \rangle_{\rho_G}$. Therefore, we can choose the gauge state as $\rho_G=\dyad{\psi_G}$ for $\ket{\psi_G}$ to be an eigenvector of $\sum_{j:Q_j Q_{k_0} \in\mathcal G}\lambda_j H'_{G,j}$ corresponding to the largest absolute eigenvalue. Note that while we only consider a single nontrivial logical component $H'_{L,k_0}$, the effective Hamiltonian $H_{\mathrm{eff}}$ remains nontrivial even if there exists another $Q_{k_1}$ with a different logical component $H'_{L,k_1}$.
    
    We can choose the logical input state $\rho_L=\dyad{\psi_L}$ for 
    \begin{equation}
        \ket{\psi_L}=\frac{1}{\sqrt{2}}(\ket{h_{\mathrm{min}}}+\ket{h_{\mathrm{max}}}),
    \end{equation}
    where $\ket{h_{\mathrm{min}}}$ and $\ket{h_{\mathrm{max}}}$ are eigenvectors of $H'_{L,k_0}$ with minimal and maximal eigenvalues, respectively. The (pure state) QFI \cite{Liu2020JPA,Sidhu2020AVSQS} of the resulting output state $\rho_\theta^{(N)}=\left(\mathcal R\circ\mathcal E_\theta\right)^N(\rho)$ is 
    \begin{equation}
        F^Q(\rho_\theta^{(N)}) = 4 N^2 \left(\mel{\psi_L}{H_{\mathrm{eff}}^2}{\psi_L} - \mel{\psi_L}{H_{\mathrm{eff}}}{\psi_L}^2 \right),
    \end{equation}
    which clearly implies the HL.
\end{proof}

\section{Proof of Theorem~\ref{thm:HL syndrome free}}
Before proving the theorem we introduce some preliminaries and peripheral results.
\subsection{Symplectic representation of Pauli operators}\label{app:symplectic representation}
Every element of the $n$-qubit Pauli group $\mathcal P_n$ can be written as
\begin{equation}
    P = \omega X^{\mathbf x} Z^{\mathbf z}
\end{equation}
for $\omega \in \{\pm 1, \pm i\}$, where $\mathbf x, \mathbf z \in \mathbb Z_2^n=\{0,1\}^n$, and we have used the notation $X^{\mathbf x} = X_1^{x_1}\cdots X_n^{x_n}$ and $Z^{\mathbf z} = Z_1^{z_1}\cdots Z_n^{z_n}$. By ignoring the phase $\omega$, we uniquely represent the $n$-qubit Pauli operator $P$ by a binary row vector \cite{Nielsen2010Quantum}
\begin{equation}
    v(P)=(\mathbf x \mid \mathbf z) \in \mathbb Z_2^{2n}.
\end{equation}
Considering $v(P_1)=(\mathbf x_1 \mid \mathbf z_1)$ and $v(P_2)=(\mathbf x_2 \mid \mathbf z_2)$, the multiplication $P_1P_2$ is translated into the binary sum of vectors (addition modulo $2$)
\begin{equation}
    v(P_1)+v(P_2)=(\mathbf x_1+\mathbf x_2 \mid \mathbf z_1+\mathbf z_2).
\end{equation}
The commutation relation is characterized by a symplectic inner product: $[P_1,P_2]=0$ if and only if \begin{equation}
    \langle v(P_1),v(P_2)\rangle_s:=v(P_1)\Lambda [v(P_2)]^T=\mathbf{x_1}\cdot\mathbf z_2+\mathbf{x_2}\cdot\mathbf z_1=0
\end{equation}
for a $2n\times 2n$ matrix
\begin{equation}
    \Lambda=\mqty(0 & I_{n\times n} \\ I_{n\times n} & 0).
\end{equation}
Similarly, $\{P_1,P_2\}=0$ if and only if $\langle v(P_1),v(P_2)\rangle_s=1$.

\subsection{Symplectic Gram-Schmidt
orthogonalization procedure} \label{sec:SGSOP}
The goal of the symplectic Gram--Schmidt orthogonalization procedure \cite{Wilde2009PRA} is to transform a set of Pauli operators $\{g_1,\dots,g_m\}$ into a canonical form $\{Z'_1,X'_1,\dots,Z_\alpha',X'_\alpha,S_1,\dots,S_\beta\}$. Here each $S_i$ for $i=1,\dots,\beta$ commutes with all the Pauli operators in the whole set; for $j=1,\dots,\beta$ each pair $(Z_j',X_j')$ anticommutes, while operators belonging to different pairs commute. 

Translated into the symplectic representation, one can apply a recursive procedure to find such a transformation. We denote by $V$ the space generated by vectors $v(g_1),\dots,v(g_m)$. Let us start with the vector $v(g_1)$. If there exists a vector $v(g_j)$ for $2\le j\le m$ such that $\langle v(g_1),v(g_j)\rangle_s=1$,
then $(v(g_1),v(g_j))$ forms a symplectic pair and can be relabeled as $(v(Z_1'),v(X_1'))$, which is extracted into a buffer. Every remaining vector $v\in V$ is replaced by
\begin{equation}
    v+\langle v,v(Z_1')\rangle_s v(X'_1)
+\langle v,v(X'_1)\rangle_s v(Z'_1),
\end{equation}
which is symplectically orthogonal to both $v(Z_1')$ and $v(X_1')$. If no such $v(g_j)$ exists, then $v(g_1)$ is labeled as the vector $v(S_1)$ and extracted into the processed buffer. Iterating this process on the remaining vectors yields the desired canonical form above.

\subsection{Existence of Clifford transformation} \label{app:lem:clifford existence}
\begin{lemma} \label{lem:clifford existence}
    Let $\{P_1,\dots,P_J\} \in \mathcal P_n$ and $\{Q_1,\dots,Q_J\} \in \mathcal P_n$ be two sets of multi-qubit Pauli operators. If they have the same commutation relations and multiplication relations, i.e.,
    \begin{equation}
        [P_i,P_j]=0 \iff [
    Q_i,Q_j]=0,\ \forall i,j=1,\dots,J
    \end{equation}
    and
    \begin{equation}
        P_iP_j=\omega P_k \iff Q_iQ_j=\omega Q_k,\ \omega\in\{\pm 1,\pm \mathrm i\},\ \forall i,j=1,\dots,J,
    \end{equation}
    then there exists a Clifford unitary $U$ such that $U P_i U^\dagger=Q_i,\ \forall i=1,\dots,J$.
\end{lemma}

\begin{proof}
We can equivalently consider an isomorphism between groups $G_1=\langle P_1,\dots,P_J \rangle$ and $G_2=\langle Q_1,\dots,Q_J \rangle$. Using the symplectic representation (see Section~\ref{app:symplectic representation} for a brief introduction), for two subspaces $V_1=\mathsf{Span}\{v(P_1),\dots,v(P_J)\}$ and $V_2=\mathsf{Span}\{v(Q_1),\dots,v(Q_J)\}$, a Clifford transformation corresponds to a linear isometry from $V_1$ to $V_2$. By Witt's theorem for symplectic spaces \cite{Artin1957Geometric}, this isometry can be extended to a linear symplectic automorphism of the global space $\mathbb Z_2^{2n}$. Returning to the Pauli group, we denote this automorphism by $a(\cdot)$ that satisfies 
\begin{equation} \label{eq:automorphism between Paulis}
    a(P_i)=Q_i,\ \forall i=1,\dots,J.
\end{equation} 
Thus, our goal is to find a Clifford transformation that implements a permutation on the whole Pauli group and coincides with Eq.~\eqref{eq:automorphism between Paulis}. Equivalently, we need to identify a Clifford unitary $U$ such that $UZ_mU^\dagger=a(Z_m)$ and $UX_mU^\dagger=a(X_m)$ for $m=1,\dots,n$.

Following the standard textbook treatment \cite{Gottesman2018IQC}, we can construct the action of $U$ on the computational basis. First, we note that
\begin{equation}
    U\ket{0}^{\otimes n}=a(Z_m)U\ket{0}^{\otimes n}, \forall m=1,\dots,n
\end{equation}
must be the normalized state vector stabilized by $a(Z_1),\dots,a(Z_n)$. Then we can define $U$ for other basis vectors
\begin{equation}
    \begin{aligned}
        U\ket{x_1,\dots,x_n}&=UX^{x_1}\cdots X^{x_n}\ket{0}^{\otimes n} \\
        &=a(X_1)^{x_1}\cdots a(X_n)^{x_n} U\ket{0}^{\otimes n} \\
        &=a(X_1^{x_1}\cdots X_n^{x_n} )U\ket{0}^{\otimes n}
    \end{aligned}
\end{equation}
for $(x_1,\dots,x_n)\in\mathbb Z_2^n$. We can verify that $U$ indeed preserves the inner product between vectors, since
\begin{equation}
    \begin{aligned}
        \bra{x_1',\dots,x'_n}U^\dagger U\ket{x_1,\dots,x_n} &=\bra{0}^{\otimes n} U^\dagger a(X_1^{x'_1+x_1}\cdots X_n^{x'_n+x_n}) U\ket{0}^{\otimes n}\\
        &= \bra{0}^{\otimes n} U^\dagger a(X_1^{x'_1+x_1}\cdots X_n^{x'_n+x_n}Z_m) U\ket{0}^{\otimes n}\\
        &=(-1)^{x'_m+x_m}\bra{0}^{\otimes n} U^\dagger a(Z_mX_1^{x'_1+x_1}\cdots X_n^{x'_n+x_n}) U\ket{0}^{\otimes n}\\
        &= (-1)^{x'_m+x_m}\bra{x_1',\dots,x'_n}U^\dagger U\ket{x_1,\dots,x_n}
    \end{aligned}
\end{equation}
for any $m=1,\dots,n$. This implies that $\bra{x_1',\dots,x'_n}U^\dagger U\ket{x_1,\dots,x_n}=0$ if and only if $\bra{x_1',\dots,x'_n}\ket{x_1,\dots,x_n}=0$.

Finally, as a sanity check, let us verify by our construction for any $\ket{x_1,\dots,x_n}$ we have
\begin{equation}
    \begin{aligned}
        UZ_mU^\dagger U\ket{x_1,\dots,x_n}&=(-1)^{x_m}U\ket{x_1,\dots,x_n}\\
        &= (-1)^{x_m}a(X_1^{x_1}\cdots X_n^{x_n} )U\ket{0}^{\otimes n}\\
        &=(-1)^{x_m}a(X_1^{x_1}\cdots X_n^{x_n} )a(Z_m)U\ket{0}^{\otimes n}\\
        &=a(Z_m)a(X_1^{x_1}\cdots X_n^{x_n})U\ket{0}^{\otimes n}\\
        &=a(Z_m)U\ket{x_1,\dots,x_n},
    \end{aligned}
\end{equation}
and
\begin{equation}
    \begin{aligned}
        UX_mU^\dagger U\ket{x_1,\dots,x_n}&=U\ket{x_1,\dots,x_m+1,\dots,x_n}\\
        &= a(X_1^{x_1}\cdots X_m^{x_m+1}\cdots X_n^{x_n} )U\ket{0}^{\otimes n}\\
        &=a(X_m)a(X_1^{x_1}\cdots X_m^{x_m} \cdots X_n^{x_n})U\ket{0}^{\otimes n}\\
        &=a(X_m)U\ket{x_1,\dots,x_n}.
    \end{aligned}
\end{equation}
Therefore, we confirm that $UZ_mU^\dagger=a(Z_m)$ and $UX_mU^\dagger=a(X_m)$ for $m=1,\dots,n$.
\end{proof}
With these tools we are ready to prove Theorem~\ref{thm:HL syndrome free}.

\subsection{Formal proof of Theorem~\ref{thm:HL syndrome free}}
\begin{proof}
    Let $\mathscr E:=\langle E_1,\dots,E_r\rangle  \subseteq \mathcal P_n$ denote the Pauli subgroup generated by the noise operators. We denote by $\mathscr E_a$ a maximal abelian subgroup of $\mathscr E$. For our purpose it suffices to ignore the phase and consider $\overline{\mathscr E}:=\mathscr E/\{\pm 1,\pm \mathrm i\}$ and $\overline{\mathscr E}_a:=\mathscr E_a/\{\pm 1,\pm \mathrm i\}$. $H\notin\mathsf{Span}\{\langle E_1,\dots,E_r \rangle\}$ implies that there exists at least one Pauli component $Q=Q_{k_0}$ in the decomposition of $H=\sum_k\lambda_kQ_{k}$ such that $Q\notin \mathsf{Span}\{\mathscr E\}$. 

    We first consider the case where $Q$ commutes with $E_i$ for all $i=1,\dots,r$. Apparently, two algebras $\mathcal A=\mathsf{Span}\{\mathscr E\}$ and $\mathcal B=\mathsf{Span}\{\langle Q\rangle\}$ commute with each other and have only trivial overlap $\mathcal A \cap \mathcal B=\mathbb CI$. It is well known that they can induce a tensor product structure $\mathcal H=\mathcal H_G\otimes\mathcal H_L$ such that all $\{E_i\}$ are operators acting nontrivially only on the gauge subsystem $\mathcal H_G$ and $Q$ acts nontrivially only on the logical subsystem $L$ \cite{Zanardi2004PRL,Poulin2005PRL}. Here we provide an intuitive constructive proof for completeness and clarity.

    It is important to note that, by Lemma~\ref{lem:clifford existence} there exists a Clifford unitary transformation $\mathcal U(\cdot)=U(\cdot)U^\dagger$ that maps a set of independent generators of the group $\overline{\mathscr E}$ to single-qubit Pauli errors (here we perform the symplectic Gram-Schmidt orthogonalization introduced in Section~\ref{sec:SGSOP} for the generators)
    \begin{equation}\label{eq:canonical noise}
    U\overline{\mathscr E}U^\dagger
    =
    \left\langle Z_1,\dots,Z_\alpha,X_1,\dots,X_\beta
    \right\rangle
    \end{equation} 
    (for $\beta\le \alpha$), and maps $Q$ to 
    \begin{equation}
        UQU^\dagger=Z_{\alpha+1},
    \end{equation}
    as the commutation relations and multiplication relations are preserved. By definition, $\alpha$ is the number of independent generators for $\overline{\mathscr E}_a$. If $\alpha\le n-1$, then qubit $\alpha+1$ is already available to serve as a logical qubit. If $\alpha=n$, we append one noiseless ancillary qubit (labeled $n+1$) on which the noise acts trivially,
    i.e. we replace each noise $E_i$ by $E_i\otimes I_{n+1}$ and view all operators as acting on $n+1$ qubits. Then \eqref{eq:canonical noise} still holds with the same $\alpha=n$, and now qubit $\alpha+1=n+1$ is available. For convenience we introduce the total number of qubits $n'$, satisfying $n'=n$ if $\alpha\le n-1$ and $n'=n+1$ if $\alpha=n$. In this commuting case $\alpha \le n-1$ naturally holds. Using this partition, all the noise operators $E_i$ act nontrivially only on $\alpha$ gauge qubits and $Q$ acts nontrivially only on one logical qubit. By choosing the gauge group
    \begin{equation}\label{eq:gauge group choice}
        \mathcal G=\langle \mathrm i, U^\dagger Z_1 U, \dots,U^\dagger Z_{\alpha} U,U^\dagger X_{1} U, \dots, U^\dagger X_{\alpha} U \rangle,
    \end{equation}
    we have $E_iE_j \notin \mathcal \mathcal P_n-\mathcal G,\ \forall i,j$ and $Q \in \mathcal P_n-G$. Therefore, if Condition~(i) holds, by Theorem~\ref{thm:subsystem code HL condition} we can achieve the HL with a syndrome-free subsystem code without additional ancillae.

    Next, let us turn to the case where $Q$ anticommutes with at least one independent generator of $\overline{\mathscr E}$. Assume without loss of generality that $Q$ anticommutes with generators $g_1,\dots,g_\gamma$ and commutes with the other ones $g_{\gamma+1},\dots,g_m$. By a new choice of generators $g_1, g_1g_2,\dots,g_1g_\gamma$, we can ensure that $Q$ only anticommutes with one generator $g_1$ and commutes with all the others. Let $\tilde g_1$ be a Pauli operator that anticommutes with $g_1$ and commutes with the remaining generators (its existence is immediate in the symplectic representation by completing to a symplectic basis \cite{Nielsen2010Quantum}). Now we choose a Clifford unitary $U$ such that
    \begin{equation}
        U g_1 U^\dagger = Z_1,\ U \tilde g_1 U^\dagger = X_1,\ U Q U^\dagger = X_1 Z_{\alpha+1},
    \end{equation}
    and $U\,\overline{\mathcal E}\,U^\dagger$ still satisfies \eqref{eq:canonical noise}. This is again guaranteed by Lemma~\ref{lem:clifford existence}, since the target set $\{Z_1,\dots,Z_\alpha,X_1,\dots,X_\beta, X_1Z_{\alpha+1}\}$ has exactly the same multiplication and commutation relations as $\{g_1,\dots,g_m,Q\}$ by construction. A recovery operation $\mathcal R$ is simply a gauge reset in the rotated frame:
    \begin{equation}
        \mathcal R(\rho)=
        U\Big(\Tr_G[U^\dagger \rho U]\otimes \rho_G\Big)U^\dagger,
    \end{equation}
    where the gauge state on the first $\alpha$ qubits can be chosen to be an eigenstate of $X_1$ for maximizing the absolute value of $\Tr(\rho_GX_1)$ in the rotated basis. Finally, if $\alpha\le n-1$, then we take $n'=n$ and no ancilla is needed, proving the achievability of HL under Condition~(ii).
\end{proof}
\end{document}